\documentclass{article}
\usepackage{wright}
\usepackage{graphicx} 
\usepackage{amsmath}

\usepackage{tikz}
\usetikzlibrary{calc,decorations.pathreplacing}

\newcommand{\qpa}[1]{\Phi_{\mathsf{QPA}}^{(#1)}}
\newcommand{\rsk}{\mathrm{RSK}}
\newcommand{\bDelta}{\boldsymbol{\Delta}}
\newcommand{\bcalE}{\boldsymbol{\calE}}
\newcommand{\ucg}{U_{\mathrm{CG}}}
\newcommand{\calL}{\mathcal{L}}

\newcommand{\schur}{U_{\mathrm{Schur}}}

\title{Nonasymptotic bounds for quantum purity amplification}
\author{Thilo Scharnhorst\thanks{UC Berkeley. \texttt{\{thilo,jspilecki,jswright\}@berkeley.edu}}  \and Jack Spilecki\footnotemark[1] \and John Wright\footnotemark[1]}
\date{}

\begin{document}

\maketitle

\begin{abstract}
    In \emph{quantum purity amplification},
    one is given $n$ copies of a noisy quantum state $\rho \in \C^{d \times d}$
    and asked to prepare $k$ copies of its principal eigenstate $\ket{v_d}$.
    Several prior works~\cite{CEM99,Fu16,LFIC24} have derived information-theoretically optimal algorithms for this problem,
    but the bounds they prove are only shown in the asymptotic regime as the number of samples $n$ tends to infinity.
    In this paper, we establish the following nonasymptotic guarantee: if $\rho$'s eigenvalues are sorted $p_1 \leq \cdots \leq p_d$ and $p_{d-1} < p_d$, then
    \begin{equation*}
        n = O\Big(k + \frac{k}{\delta} \cdot \frac{1-p_d}{(p_d-p_{d-1})^2}\Big)
    \end{equation*}
    copies suffice to output a state with fidelity at least $1-\delta$ with $\ket*{v_d^{\otimes k}}$. Our bound holds for arbitrary spectra, and is independent of the dimension $d$. In the case of depolarizing noise, our finite-sample guarantee matches the optimal asymptotic scaling. Our proof is based on the combinatorics of random Young diagrams.
\end{abstract}

\newpage



\section{Introduction}

One of the most difficult challenges in the design and implementation of quantum computers is the presence of noise in quantum systems.
Many quantum algorithms and protocols rely on the creation of delicately prepared pure quantum states, and imperfections in the quantum computer, as well as interactions with the environment, can perturb or even destroy these states.
Several approaches for managing this error have been suggested in the literature, such as quantum fault tolerance and quantum error mitigation.
We study a third approach known as \emph{quantum purity amplification}.\footnote{This task is also commonly referred to in the literature as \emph{quantum state purification}. We use the term \emph{quantum purity amplification} to avoid confusion with the notion of a purification of a mixed state.}
To define this approach, let us first fix some notation:
throughout this work, $\rho \in \C^{d \times d}$ will denote a $d$-dimensional mixed state with eigendecomposition
\begin{equation*}
    \rho = \sum_{i=1}^d p_i \cdot \ketbra{v_i},
\end{equation*}
with the eigenvalues sorted as $p_1 \leq \cdots \leq p_d$.
We will always assume $\rho$ has a unique principal eigenstate, i.e.\ that $p_{d-1} < p_d$.
Intuitively, we will imagine that $\ket{v_d}$ is a pure state we want access to, but it has been sent through some sort of noisy channel, resulting in the mixed state $\rho$.
In the problem of quantum purity amplification, we are given access to $n$ copies of $\rho$, and the goal is to produce a quantum state whose fidelity with $\ket{v_d}$ is high.
We will also consider the variant in which we want to recover $k$ copies of $\ket{v_d}$, in which case we want a state with high fidelity to $\ket{v_d^{\otimes k}}$.

Quantum purity amplification was first studied in a work of Cirac, Ekert, and Macchiavello~\cite{CEM99},
although they were motivated by related problems in the topic of entanglement purification~\cite{BBP+96}.
Their main result was to identify and analyze the optimal algorithm for purifying qubit states.
They showed that their algorithm, when given $n$ copies of a mixed qubit state $\rho = p_1 \cdot \ketbra{v_1} + p_2 \cdot \ketbra{v_2}$ with $p_2 > p_1$, outputs a mixed state $\sigma$ with fidelity
\begin{equation*}
    \bra{v_2} \cdot \sigma \cdot \ket{v_2} = 1 - \frac{1}{n} \cdot \frac{p_1}{(p_2 - p_1)^2} + O_p\Big(\frac{1}{n^2}\Big).
\end{equation*}
Note that the final error term $O_p(1/n^2)$ is allowed to depend arbitrarily on the probability vector $p = (p_1, p_2)$.
Thus, such a bound cannot give any concrete guarantees for the output fidelity of the algorithm when the number of copies $n$ is a reasonable, finite number (such as $n = 1,000,000$); instead, it can only provide guarantees in the \emph{asymptotic regime} when $n \rightarrow \infty$.
In this regime their bound says that to achieve output fidelity $\bra{v_2} \sigma \ket{v_2} \geq 1 - \delta$ in the limit as $\delta \rightarrow 0$,
\begin{equation*}
    n = O\bigg(\frac{1}{\delta} \cdot \frac{1-p_2}{(p_2 - p_1)^2}\bigg)
\end{equation*} 
copies of $\rho$ suffice.
Hence, the number of samples blows up when $p_1$ approaches $p_2$,
as in this regime it is difficult to distinguish $\ket{v_2}$ from the surrounding noise,
but $n$ becomes reasonable once $p_1$ and $p_2$ are gapped.

Since this work, it has been an open problem to devise the optimal quantum purity amplification algorithm for qudit rather than qubit states.
The chief difficulty is that the algorithm of~\cite{CEM99} is heavily representation theoretic, and the representation theory of quantum states, already challenging in the qubit setting, becomes far more difficult in the qudit setting.
A first step in this direction was made by Fu in his master's thesis~\cite{Fu16}, where he identified and analyzed the optimal algorithm for purifying qutrit states in the special case of depolarizing noise.
This is the physically-motivated case when $p_1 = \cdots = p_{d-1} < p_d$, and the goal is to output $\ket{v_d}$.
In this case, we can view the state $\rho$ as being the result of passing $\ket{v_d}$ through the depolarizing channel
\begin{equation*}
    \Phi_{\eta}(\ketbra{v_d}) = (1- \eta) \cdot \ketbra{v_d} + \eta \cdot (I/d),
\end{equation*}
for some noise rate $\eta$.
Assuming this symmetry among the smaller eigenvalues makes the representation theory more tractable, thereby making it possible to identify the optimal channel.
As in~\cite{CEM99}, Fu is able to analyze the output fidelity of this optimal qutrit channel, but only in the case of asymptotically large $n$.

Recently, the work of Li, Fu, Isogawa, Silva, and Chuang~\cite{LFI+25} has identified the optimal algorithm for purifying qudit states in the case of depolarizing noise.
In the asymptotic limit as $\delta \rightarrow 0$, the output of their algorithm achieves fidelity $\bra{v_d} \sigma \ket{v_d} \geq 1 - \delta$ once
\begin{equation}\label{eq:qpa-bound-depolarizing}
    n = O\bigg(\frac{1}{\delta} \cdot \frac{1-p_d}{(p_d - p_{d-1})^2}\bigg),
\end{equation}
which generalizes the qubit and qutrit results of~\cite{CEM99} and~\cite{Fu16}, respectively.
They also give an algorithm for purifying general, not necessarily depolarized, qudit states.
Although they do not show that this algorithm is optimal for this task, they are still able to analyze it and show that in the asymptotic regime of $\delta \rightarrow 0$, it achieves fidelity $1 - \delta$ once
\begin{equation}\label{eq:full-spectrum}
    n = O\bigg(\frac{1}{\delta} \cdot \Big(\frac{p_1}{(p_d - p_1)^2} + \cdots + \frac{p_{d-1}}{(p_d - p_{d-1})^2}\Big)\bigg).
\end{equation}
Note that $(p_d-p_i)^{-2} \leq (p_d - p_{d-1})^{-2}$ for all $1 \leq i \leq d-1$ and that $p_1 + \cdots + p_{d-1} = 1 - p_d$.
Applying these two facts, we have that
\begin{equation*}
    \eqref{eq:full-spectrum} \leq O\bigg(\frac{1}{\delta} \cdot \frac{1-p_d}{(p_d - p_{d-1})^2}\bigg).
\end{equation*}
As a result, \eqref{eq:full-spectrum} generalizes their bound from~\eqref{eq:qpa-bound-depolarizing}, and even improves upon it when $p$ has a lot of weight spread out among eigenvalues which are significantly smaller than $p_{d-1}$.

In practice, however, we do not have access to an asymptotic number of copies of $\rho$, but only a concrete number $n$ of them.
Motivated by this, we study the problem of proving \emph{nonasymptotic} bounds for quantum purity amplification, which hold for any error parameter $\delta > 0$ and not just infinitesimally small $\delta$.
Our main result gives a nonasymptotic bound which matches the asymptotic bound in \eqref{eq:qpa-bound-depolarizing} for the depolarizing channel, generalizes it to arbitrary spectra, and extends it to the case of $k$-copy outputs.

\begin{theorem}[Nonasymptotic bounds for quantum purity amplification]\label{thm:nonasymptotic-bounds}
    There exists an absolute constant $C$ such that the following is true.
    Let $\delta > 0$ and $k \geq 1$ be an integer.
    Then for any mixed state $\rho \in \C^{d \times d}$,
    we have that
    \begin{equation*}
        \bra{v_d^{\otimes k}} \qpa{k}(\rho^{\otimes n}) \ket{v_d^{\otimes k}} \geq 1 - \delta,
        \quad
    \text{so long as}\quad
        n \geq C \cdot k + C \cdot \frac{k}{\delta} \cdot \frac{1 - p_d}{(p_d-p_{d-1})^2}.
    \end{equation*}
\end{theorem}

Let us note that the first $O(k)$ term in our sample complexity is needed simply to ensure that at least $O(k)$ copies of $\rho$ are given as input, as this is surely required if we hope to produce $k$ copies of $\ket{v_d}$.
(Note that the second term can actually be smaller than $k$ if $p_d$ is very close to 1.)
In most cases, however, the second term will be the one which dominates, and we can see that it matches the bound in \eqref{eq:qpa-bound-depolarizing}.
Our techniques are based on the combinatorics of Young diagrams and longest increasing subsequences.
These techniques have been previously used to analyze quantum state learning and testing algorithms in works such as~\cite{OW15,OW16,OW17a}, and we show that they can be used to analyze quantum purity amplification algorithms as well.
Our bound leaves open the question of whether a fine-grained version of the bound in~\eqref{eq:full-spectrum} can be shown in the nonasymptotic setting.

\subsection{Related work}

Quantum purity amplification has been studied in a variety of settings,  yielding a suite of algorithms with different tradeoffs.
The representation theoretic algorithms we have discussed so far, and which we study in this work, are able to achieve the information-theoretic optimal sample complexities.
However, they tend to be computationally more challenging to implement, and they require having access to all $n$ copies of $\rho$ up front.
An elegant alternative approach based on a tree of $\mathrm{SWAP}$ tests was proposed by Childs, Fu, Leung, Li, Ozols, and Vyas~\cite{CFL+25}.
Their algorithm operates in the \emph{streaming} model, in which the copies of $\rho$ come in one-at-a-time, and only requires $O(\log(n))$ qudits of memory.
They consider the case of depolarizing noise, and show a nonasymptotic bound which achieves the optimal bound from~\eqref{eq:qpa-bound-depolarizing} up to constant factors, at least in the case when $p_d \geq 2/3 + 1/(3d)$;
when $p_d$ is smaller, they show a bound of $n = O_p(1/\delta)$.
Following their work, Grier, Leung, Li, Pashayan, and Schaeffer~\cite{GLL+25} extended the $n = O_p(1/\delta)$ sample bound to hold against all states $\rho$, not just in the case of depolarizing noise.

Another approach for quantum purity amplification comes from the work of Dalzell,  Gilyén,  Hann,  McArdle, Salton,  Nguyen,  Kubica,  and Brandão~\cite{DGH+25}.
Interestingly, their motivation came from designing a protocol for implementing a fault tolerant quantum random access memory.
For their application, they also wanted an algorithm which works in the streaming model, but they also wanted strong bounds in the regime where $p_d \ll 1$. 
Their main result pertaining to quantum purity amplification is a streaming algorithm which uses
\begin{equation*}
    n = O\bigg(\Big(\frac{1}{\delta} + \frac{1}{p_d}\Big) \cdot \frac{1-p_d}{(p_d - p_{d-1})^2}\bigg),
\end{equation*}
copies of $\rho$, nearly matching the optimal bound in the case of depolarizing noise, all while maintaining only $O(1)$ qudits of memory.
Their algorithm is based on the quantum singular value transform framework.

\paragraph{Concurrent and independent work.}
We recently learned of the independent work of Li, Theil, Harrow, and Chuang~\cite{LTHC26a}.
They identified the optimal algorithm for purifying arbitrary qudit states in the case of general spectra,
and generalized it further to the problem of extracting the eigenstate corresponding to the $j$-th largest eigenvalue of $\rho$ rather than just the first.
Finally, they are also able to derive nonasymptotic bounds on the sample complexities of their algorithm.
As observed in their follow-up work~\cite[Equation (4)]{LTHC26b}, 
their bounds are strong enough to imply a nonasymptotic sample complexity of
\begin{equation}\label{eq:lthc-bound}
    n = O\bigg(\frac{k}{\delta}\cdot \frac{1}{(p_d-p_{d-1})^2}\bigg),
\end{equation}
which essentially matches our bound for general spectra.
In addition, they are also able to derive more fine-grained fidelity bounds which are similar in spirit to the bound in equation \eqref{eq:full-spectrum} (see, e.g.\ their \cite[Lemma S6.36]{LTHC26a}).

\section*{Acknowledgments}

We thank András Gilyén and Zhaoyi Li for helpful discussions, and Angelos Pelecanos for collaboration during the early stages of this project.
J.S.\ and J.W.\ are supported by the NSF CAREER award CCF-233971.

\section{Preliminaries}

We use \textbf{boldface} to denote random variables.
We write $[n]$ for the set of integers $\{1, 2, \dots, n\}$. 
The \emph{falling factorial} is denoted as $n^{\downarrow k} \coloneqq n (n-1) \cdots (n - k+1)$. 
Given a pure state $\ket{\psi}$, we denote its $k$-fold tensor power either by $\ket{\psi}^{\otimes k}$ or by $\ket*{\psi^{\otimes k}}$.
A tuple of numbers $a = (a_1, \ldots, a_n)$ is \emph{weakly increasing} if $a_1 \leq \cdots \leq a_n$ and \emph{strongly increasing} if $a_1 < \cdots < a_n$.
A \emph{subsequence of $a$} is a string $(a_{i_1}, \ldots, a_{i_k})$, where $1 \leq i_1 < \cdots < i_k \leq n$.

\subsection{Young diagrams and Young tableaux}
In the next few sections, we will introduce the combinatorics of random words and its relationship to the representation theory which arises when studying quantum states.
For a more thorough introduction to these topics, see~\cite{Wri16}.

A \emph{partition of $n$}, denoted $\lambda \vdash n$, is a tuple of integers $\lambda = (\lambda_1, \ldots, \lambda_d)$ satisfying $\lambda_1 + \cdots + \lambda_d = n$ and $\lambda_1 \geq \cdots \geq \lambda_d \geq 0$.
We write $\ell(\lambda)$ for the \emph{length} of $\lambda$, which is the number of its nonzero entries.
Partitions $\lambda$ are often visually represented as \emph{Young diagrams}, which consist of boxes arranged into rows of lengths $\lambda_1$, $\lambda_2$, and so forth.
A \emph{standard Young tableau (SYT)} $S$ is a Young diagram whose boxes have been filled with the numbers $1, \ldots, n$, with the property that each row is strongly increasing when read from left-to-right and each column is strongly increasing when read from top-to-bottom.
A \emph{semistandard Young tableau (SSYT) $T$ with alphabet $[d]$} is a Young diagram whose boxes have been filled with numbers in $[d]$, with the property that each row is weakly increasing when read from left-to-right and each column is strongly increasing when read from top-to-bottom.
We illustrate these in Figure~\ref{fig:example_Young_tableaux}. 

\begin{figure}[h!]
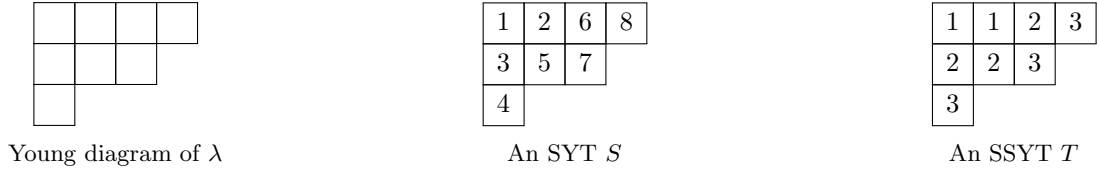

\centering

\begin{minipage}[b]{0.28\textwidth}
\centering
\ydiagram{4,3,1}\\[0.6em]
{\small Young diagram of $\lambda$}
\end{minipage}
\hfill
\begin{minipage}[b]{0.28\textwidth}
\centering
\begin{ytableau}
1 & 2 & 6 & 8 \\
3 & 5 & 7 \\
4
\end{ytableau}\\[0.6em]
{\small An SYT $S$}
\end{minipage}
\hfill
\begin{minipage}[b]{0.28\textwidth}
\centering
\begin{ytableau}
1 & 1 & 2 & 3 \\
2 & 2 & 3 \\
3
\end{ytableau}\\[0.6em]
{\small An SSYT $T$}
\end{minipage}

\caption{From left to right: a Young diagram of shape $\lambda=(4,3,1)$, a standard Young tableau (SYT) $S$ of shape $\lambda$, and a semistandard Young tableau (SSYT) $T$ of shape $\lambda$ with alphabet $[3]$.}
\label{fig:example_Young_tableaux}
\end{figure}

The shape $\shape(S)$ of an SYT $S$ is the Young diagram that results from erasing all the entries of $S$; the shape $\shape(T)$ of an SSYT $T$ is defined similarly.
Given $1 \leq i \leq d$,
we will write $T^{< i}$ for the SSYT produced by removing any box containing a letter which is equal to $i$ or larger.
The \emph{type} of $T$ is specified by a histogram $h = (h_1, \ldots, h_d)$, in which $h_i$ equals the number of occurrences of the letter $i$ in $T$.

A \emph{word} is a tuple $w = (w_1, \ldots, w_n) \in [d]^n$.
The \emph{type} of $w$ is specified by a histogram $h = (h_1, \ldots, h_d)$, in which $h_i$ equals the number of occurrences of the letter $i$ in $w$.
The well-known \emph{Robinson--Schensted--Knuth (RSK) correspondence} establishes a bijection between words $w \in [d]^n$ and pairs $(S, T)$ with the following properties: (i) $S$ is an SYT, (ii) $T$ is an SSYT, (iii) $S$ and $T$ have the same shape $\lambda$, and (iv) $T$ has the same type as $w$.
The algorithm which computes this bijection is known as the \emph{RSK algorithm}, and we write $\rsk(w) = (\lambda, S, T)$ for the output of this bijection; although $\lambda$ is implicit in $S$ and $T$, we include it for convenience.
An observation first made by Schensted~\cite{Sch61}, and later generalized by Greene~\cite{Gre74},
is that $\lambda_1$ is equal to the length of the longest weakly increasing subsequence in $w$.
A consequence is that $\lambda_1 \geq h_d$, as the all-$d$'s subsequence is a weakly increasing subsequence in $w$.

Given a probability distribution $p = (p_1, \ldots, p_d)$, we write $\bw \sim p^{\otimes n}$ for a word $\bw = (\bw_1, \ldots, \bw_n) \in [d]^n$ in which each letter $\bw_i$ is sampled independently from $p$.
A sample from the \emph{RSK distribution} $\rsk^n(p)$ is distributed as $(\blambda, \bS, \bT) = \mathrm{RSK}(\bw)$, where $\bw \sim p^{\otimes n}$.
Note that if $w \in[d]^n$ is a fixed word and $(\lambda, S, T) = \rsk(w)$, then because the RSK algorithm computes a bijection, we have that
\begin{equation*}
    \Pr[(\blambda, \bS, \bT) = (\lambda, S, T)]
    = \Pr[\bw = w] = p_1^{h_1} \cdots p_d^{h_d} \eqqcolon p^h,
\end{equation*}
where $h$ is the type of $w$.
Thus, since $T$ and $w$ have the same type, if we set $p^T \coloneqq p^h$, we have that
\begin{equation}\label{eq:rsk-pdf}
\Pr[(\blambda, \bS, \bT) = (\lambda, S, T)] = p^T.
\end{equation}

The distribution of the Young diagram $\blambda$ has been extensively studied in both the mathematics literature~\cite{KV86,TW01,Joh01,ITW01,Bia01,Mel10a,Mel10b,HX13,Mel12} and the quantum computing literature~\cite{ARS88,KW01,Kup02,HM02,CM06,CHW07,OW15,OW16,OW17a}.
As we have seen, the length of the first row $\blambda_1$ is at least $\bh_d$, where $\bh$ is the histogram of $\bw$, and so $\E[\blambda_1] \geq \E[\bh_d] = p_d \cdot n$.
Of course, the longest increasing subsequence of $\bw$ might also use letters which are not $d$ to attain a length which is longer than $\bh_d$.
However, because $p_d$ is the largest probability value, these additional letters turn out to not help too much, and for large $n$, it turns out that $\blambda_1$ tends to concentrate around $p_d \cdot n$.
Most strikingly, it turns out that if there is a gap between $p_d$ and the other probability values, then the non-$d$ letters can only increase the expected length of the longest increasing subsequence by at most an absolute constant. 
\begin{lemma}[$k=1$ case of {\cite[Theorem 1.13]{OW17a}}]\label{lem:first-row}
For all $n$,
\begin{equation*}
    \E[\blambda_1] \leq p_d \cdot n + \sum_{i=1}^{d-1} \frac{p_i}{p_d - p_i}.
\end{equation*}
\end{lemma}
\noindent 
This bound turns out to be tight, and indeed the two sides are actually equal to each other in the limit of large $n$~\cite{ITW01}.
Like the first row, the length of the second row $\blambda_2$ tends to concentrate around $p_{d-1} \cdot n$, the length of the third row $\blambda_3$ tends to concentrate around $p_{d-2} \cdot n$, and so forth. We will need the following result along these lines.
\begin{lemma}[$k = 2$ case of {\cite[Theorem 4.12]{OW17a}}]\label{lem:second-row}
$
    \E[
    (\blambda_2-p_{d-1} \cdot n)^2]
    \leq
    84 p_{d-1} \cdot n+42 (1 - p_{d}) \cdot n.
$
\end{lemma}

If a tuple $\lambda = (\lambda_1, \ldots, \lambda_d)$ satisfies $\lambda_1 \geq \cdots \geq \lambda_d$ but not necessarily $\lambda_d \geq 0$, then we call it a \emph{staircase}.
A staircase can be viewed as a Young diagram which allows negative boxes.
\ignore{
We will need to generalize SSYTs to the case of negative boxes as well.
To do so, we will first consider an alternative representation for a traditional SSYT $T$.
For $1 \leq i \leq d$, write $T^i = (T^i_1, \ldots, T^i_i)$ for the Young diagram consisting of all boxes from $T$ containing letters whose values are $\leq i$. These Young diagrams satisfy the \emph{interlacing condition} $T^i \precsim T^{i+1}$, which means that
\begin{equation*}
    T^{i+1}_1 \geq T^i_1 \geq T^{i+1}_2 \geq T^i_2 \geq \cdots \geq T^i_i \geq T^{i+1}_{i+1}.
\end{equation*}
Indeed, any tuple of Young diagrams $(T^1, \ldots, T^d)$ which satisfies the interlacing conditions $T^i \precsim T^{i+1}$ corresponds to an SSYT $T$, and vice versa. We illustrate this in Figure~\ref{fig:example_Young_tableaux}.
Later, we will consider the more general case where this tuple consists of staircases (rather than Young diagrams) which satisfy the interlacing conditions; in this case, we will refer to it as a \emph{Gelfand-Tsetlin (GT) pattern}.}

\subsection{Representation theory of the symmetric and unitary groups}

We write $S_n$ for the symmetric group of permutations acting on $n$ items.
The irreducible representations (irreps) of $S_n$ are indexed by Young diagrams $\lambda \vdash n$ and are written $(P_{\lambda}, p_{\lambda})$, where $P_{\lambda}$ is referred to as the \emph{Specht module}.
There is an orthonormal basis of $P_{\lambda}$ with one vector $\ket{S}$ for each SYT $S$ consisting of $n$ boxes.
We write $\dim(\lambda)$ for the dimension of $P_{\lambda}$, which is therefore equal to the number of SYTs with $n$ boxes. 

We write $GL(d)$ for the general linear group consisting of all invertible $d \times d$ complex matrices.
The polynomial irreps of $GL(d)$ are indexed by Young diagrams $\lambda$ with $\ell(\lambda) \leq d$ and are written $(Q_{\lambda}^d, q_{\lambda}^d)$, where $Q_{\lambda}^d$ is referred to as the \emph{Weyl module}.
They are called polynomial irreps because given $M \in GL(d)$, the matrix entries of $q_{\lambda}^d(M)$ are polynomials in the matrix entries of $M$.
There is an orthonormal basis of $Q_{\lambda}^d$ with one vector $\ket{T}$ for each SSYT $T$ with alphabet $[d]$.
When the inputs are restricted to the unitary group $U(d)$, then $(Q_{\lambda}^d, q_{\lambda}^d)$ is also an irrep of $U(d)$. 

There is a natural representation $P(\cdot)$ of $S_n$ on the space $(\C^d)^{\otimes n}$ which acts as follows: for each $\pi \in S_n$,
\begin{equation*}
    P(\pi) \cdot \ket{i_1} \otimes \cdots \otimes \ket{i_n} = \ket{i_{\pi^{-1}(1)}} \otimes \cdots \otimes \ket{i_{\pi^{-1}(n)}}, \qquad \text{for all $i_1, \ldots, i_n \in [d]$}.
\end{equation*}
Similarly, there is a natural representation $Q(\cdot)$ which acts on this space as follows: for every $M \in GL(d)$,
\begin{equation*}
    Q(M) \cdot \ket{i_1} \otimes \cdots \otimes \ket{i_n} = (M \ket{i_1}) \otimes \cdots \otimes (M \ket{i_n}), \qquad \text{for all $i_1, \ldots, i_n \in [d]$}.
\end{equation*}
These two representations commute, i.e.\ $P(\pi) \cdot Q(M) = Q(M) \cdot P(\pi)$, and so their joint action on $(\C^d)^{\otimes n}$ is well-defined.
A powerful result from representation theory known as \emph{Schur--Weyl duality} states that this joint action has an elegant decomposition into the irreps of the two groups.
Formally, this result states that there is a unitary $\schur$ acting on $(\C^d)^{\otimes n}$ such that
\begin{equation}\label{eq:Schur--Weyl}
    \schur \cdot P(\pi) Q(M) \cdot \schur^{\dagger} = \bigoplus_{\substack{\lambda \vdash n \\ \ell(\lambda)\leq d}} p_{\lambda}(\pi) \otimes q_{\lambda}^d(M).
\end{equation}
The orthonormal basis that results from applying this unitary is known as the \emph{Schur basis}, and it consists of vectors of the form $\ket{\lambda, S, T}$, where $\lambda$ is a Young diagram satisfying $\lambda \vdash n$ and $\ell(\lambda) \leq d$, $S$ is an SYT of shape $\lambda$, and $T$ is an SSYT of shape $\lambda$ and alphabet $[d]$.
These vectors have the property that $\ket{\lambda, S, T}$ is a superposition over standard basis vectors $\ket{w}$, where $w \in [d]^n$, in which $w$ and $T$ have the same type.

In this work, we study inputs of the form $Q(\rho) = \rho^{\otimes n}$.
Applying \eqref{eq:Schur--Weyl} to the case of $\pi = I$ and $M = \rho$, Schur--Weyl duality states that
\begin{equation} \label{eq:Schur--Weyl-states}
    \schur \cdot \rho^{\otimes n} \cdot \schur^{\dagger} = \bigoplus_{\substack{\lambda \vdash n \\ \ell(\lambda)\leq d}} I_{\dim(P_\lambda)} \otimes q_{\lambda}^d(\rho).
\end{equation}
Note that $\rho$ is only an element of $GL(d)$ when all of its eigenvalues are nonzero. However, this expression is still well-defined even in the case when $\rho$ has eigenvalues which are equal to 0. This is because $q_{\lambda}^d(\cdot)$ is a polynomial irrep of $GL(d)$, and so the entries of $q_{\lambda}^d(\rho)$ are polynomials in the entries of $\rho$. 
Since $\rho^{\otimes n}$ is block diagonal in the Schur basis, it is natural to measure it according to the projective measurement $\{\Pi_{\lambda}\}_{\lambda \vdash n, \ell(\lambda) \leq d}$, in which $\Pi_{\lambda}$ projects onto the $\lambda$-irrep space.
Doing so is known as \emph{weak Schur sampling}.

Let us suppose further that $\rho$ is diagonal in the standard basis, so that we can write
\begin{equation*}
    \rho = \sum_{i=1}^d p_i \cdot \ketbra{i}.
\end{equation*}
Then for a word $w \in [d]^n$ with type $h$, the standard basis vector $\ket{w}$ is an eigenvector of $\rho^{\otimes n}$ with eigenvalue $p^h$.
This means that if $T$ is an SSYT of type $h$, then $\ket{\lambda, S, T}$ is also an eigenvector of $\rho^{\otimes n}$ with eigenvalue $p^T$, because $\ket{\lambda, S, T}$ is a superposition over standard basis vectors $\ket{w}$ of the same type as $T$.
Thus, if we measure $\rho^{\otimes n}$ in the Schur basis, we will receive $\ket{\blambda, \bS, \bT}$ with probability $p^{\bT}$, which means that the $(\blambda, \bS, \bT)$ we observe is distributed as a sample from the RSK distribution $\rsk^n(p)$, due to \eqref{eq:rsk-pdf}. Thus,
\begin{equation}\label{eq:schur-rsk}
    \schur \cdot \rho^{\otimes n} \cdot \schur^{\dagger}
    = \sum_{\lambda, S, T} p^T \cdot \ketbra{\lambda, S, T}
    = \E_{(\blambda, \bS, \bT) \sim \rsk^n(p)} \big[\ketbra{\blambda, \bS, \bT}\big].
\end{equation}

Later, we will consider the 
\emph{rational} irreps of $GL(d)$, whose matrix entries are \emph{rational} functions in the matrix entries of their inputs.
These are indexed by staircases $\lambda$ and also written $(Q_{\lambda}^d, q_{\lambda}^d)$.

\subsection{The dual Clebsch--Gordan transform}

We will now survey the dual Clebsch--Gordan transform. For an excellent introduction to this topic, see~\cite{Ngu24}. 
The dual of the defining representation $Q_{(1)}^d\cong \C^d$ is denoted by $(Q_{(1)}^d)^*$, and has
standard basis
\begin{equation*}
    \ket{1^*},\ldots,\ket{d^*}.
\end{equation*}
Given a matrix $U \in U(d)$, it acts on $(Q_{(1)}^d)^*$ as $\overline{U}$.
For every staircase $\lambda$, the tensor product $(Q_{(1)}^d)^*\otimes Q_\lambda^d$ is not necessarily itself irreducible, and it therefore decomposes further into 
irreps of $GL(d)$.
This decomposition is known as the \emph{dual Clebsch--Gordan transform}, and it takes the form
\begin{equation}\label{eq:cg-transform}
    (Q_{(1)}^d)^*\otimes Q_\lambda^d
    \cong\bigoplus_{\substack{i:\lambda-e_i\\\text{is a staircase}}} Q_{\lambda-e_i}^d.
\end{equation}
Here, $\lambda - e_i$ is the staircase obtained by subtracting $1$ from the $i$-th component of $\lambda$. Pictorially, one can think of $\lambda-e_i$ as the staircase obtained from $\lambda$ by removing a box from the $i$-th row. We note that this decomposition is \emph{multiplicity-free},
meaning that each irrep of $GL(d)$ occurs on the right-hand side either zero or one times.
We write $\ucg^{\lambda}$ for the unitary which implements this decomposition. For every $U \in U(d)$, it satisfies
\begin{equation}\label{eq:cg}
    \ucg^{\lambda} \cdot
    \big(\overline{U}\otimes q_\lambda^d(U)\big) \cdot \bigl(\ucg^{\lambda}\bigr)^\dagger = \bigoplus_{\substack{i:\lambda-e_i\\\text{is a staircase}}} q_{\lambda-e_i}^d(U).
\end{equation}
Thus, the Clebsch--Gordan transform
branches the space into components corresponding to each way of removing a single box from $\lambda$. 

More generally, we will be interested in tensoring $k$ copies of $(Q^d_{(1)})^*$ onto $Q_{\lambda}^d$.
To understand how this decomposes into irreps,
let us first consider a tuple $r = (r_1, \ldots, r_k) \in [d]^k$.
We can view this as specifying a sequence of staircases; in particular, the $t$-th staircase is given by
\begin{equation*}
    \lambda - \sum_{i=1}^t e_{r_i},
\end{equation*}
and it is derived from the $(t-1)$-st staircase in the path by removing the final box in row $r_t$ (with the zeroth staircase taken to be $\lambda$ itself). Not all tuples $r$ give rise to valid staircases, for example, if $\lambda_i = \lambda_{i+1}$, and we remove a box from row $i$.
We call $r$ a \emph{$k$-path} if the $t$-th staircase is valid for $t = 0, 1, \dots, k$.
The $t=k$ endpoint of this path we denote by $\lambda(r)$.
Then we claim that the following decomposition is true:
\begin{equation}\label{eq:kfold-cg}
\big((Q_{(1)}^d)^*\big)^{\otimes k}\otimes Q_\lambda^d \cong \bigoplus_{k\text{-paths } r}  Q_{\lambda(r)}^d.
\end{equation}
Let us show this by induction.
For the base case, note that the $k=1$ case is equivalent to \eqref{eq:cg-transform}.
Next, let us assume it holds for the value of $k$.
Then by the inductive hypothesis
\begin{equation*}
\big((Q_{(1)}^d)^*\big)^{\otimes (k+1)}\otimes Q_\lambda^d 
 \cong (Q_{(1)}^d)^* \otimes \big((Q_{(1)}^d)^*\big)^{\otimes k}\otimes Q_\lambda^d \cong (Q_{(1)}^d)^* \otimes \Big(\bigoplus_{k\text{-paths } r}  Q_{\lambda(r)}^d\Big)
\cong  \bigoplus_{k\text{-paths } r}  (Q_{(1)}^d)^* \otimes Q_{\lambda(r)}^d.
\end{equation*}
Next, applying the Clebsch--Gordan transform, we have
\begin{equation}\label{eq:block-cg}
    \bigoplus_{k\text{-paths } r}  (Q_{(1)}^d)^* \otimes Q_{\lambda(r)}^d
    \cong 
    \bigoplus_{k\text{-paths } r}\Big(\bigoplus_{\substack{i:\lambda(r)-e_i\\\text{is a staircase}}} Q_{\lambda(r)-e_i}^d\Big)
    \cong \bigoplus_{(k+1)\text{-paths } r'}  Q_{\lambda(r')}^d.
\end{equation}
This completes the proof.

The unitary which implements the transformation in \eqref{eq:block-cg} we denote by $\ucg^{\lambda, k \rightarrow k+1}$; it can be written as
\begin{equation}\label{eq:one-step-cg}
    \ucg^{\lambda, k \rightarrow k+1} : \bigoplus_{k\text{-paths } r}  (Q_{(1)}^d)^* \otimes Q_{\lambda(r)}^d \rightarrow \bigoplus_{(k+1)\text{-paths } r'}  Q_{\lambda(r')}^d,
    \quad \text{where}
    \quad
    \ucg^{\lambda, k \rightarrow k+1} = \bigoplus_{k\text{-paths } r} \ucg^{\lambda(r)}.
\end{equation}
For convenience, we will also set $\ucg^{\lambda, 0 \rightarrow 1} \coloneqq \ucg^{\lambda}$.
In addition, the unitary which implements the overall transform in~\eqref{eq:kfold-cg} we denote by 
\begin{equation}\label{eq:k-fold}
    \ucg^{\lambda,k}: \big((Q_{(1)}^d)^*\big)^{\otimes k}\otimes Q_\lambda^d \rightarrow\bigoplus_{k\text{-paths } r} Q_{\lambda(r)}^d.
\end{equation}
From the above inductive proof, we have the equality
\begin{equation}\label{eq:induction}
    \ucg^{\lambda,k} = \ucg^{\lambda, k-1 \rightarrow k}\cdots \ucg^{\lambda, 1 \rightarrow 2}\ucg^{\lambda, 0 \rightarrow 1}.
\end{equation}
Rewriting, the transformation satisfies, for all $U\in U(d)$,
\begin{equation*}
    \ucg^{\lambda,k}\cdot \big(\overline{U}^{\otimes k}\otimes q_\lambda^d(U) \big)\cdot (\ucg^{\lambda,k})^\dagger
    =\bigoplus_{k\text{-paths } r} q_{\lambda(r)}^d(U).
\end{equation*}
We note that this decomposition is no longer multiplicity-free,
as multiple paths may have the same staircase as their endpoint.
The special case we will care most about is when the starting staircase $\lambda$ is a Young diagram satisfying $\lambda_1 - \lambda_2 \geq k$;
in this case, the irrep $Q_{\lambda - k e_1}^d$ \emph{does} appear on the right-hand side of \eqref{eq:k-fold} with multiplicity one, corresponding to the unique path $r = (1, \ldots, 1)$.

Explicit formulas for the matrix entries of $\ucg^{\lambda}$ in a given basis are known as \emph{Clebsch--Gordan coefficients}.
We will need a particular set of Clebsch--Gordan coefficients, and to motivate these formulas, let us first state the explicit bases that we will work in.
Let $\lambda$ be a Young diagram, and suppose $\lambda_1 - \lambda_2 \geq 1$,
so that $\lambda - e_1$ is also a Young diagram.
Via \eqref{eq:cg},
we see that the input to $\ucg^{\lambda}$ has a basis given by $\ket{i^*} \otimes \ket{T}$, where $1 \leq i \leq d$ and $T$ is an SSYT of shape $\lambda$ and alphabet $[d]$,
and the output has a basis given by $\ket{\lambda - e_i, T'}$, where $T'$ is an SSYT of shape $\lambda - e_i$ and alphabet $[d]$.

\begin{lemma}[{\cite[Sec.~18.2.10, Eq.~(5)]{VK92}}]\label{lem:cg-coeff} 
    Let $\lambda \vdash n$ be a Young diagram such that $\lambda_1 - \lambda_2 \geq 1$, and write $\lambda' = \lambda - e_1$.
    Let $T$ be an SSYT of shape $\lambda$ and alphabet $[d]$.
    Then if $T'$ is an SSYT of shape $\lambda'$ and alphabet $[d]$,
    \begin{equation*}
        \bra{\lambda', T'} \cdot \ucg^\lambda \cdot \ket{d^*}\ket{T} = 0
    \end{equation*}
    unless the final box in the first row of $T$ contains the letter ``$d$'', and $T'$ is the SSYT which results from removing this box.
    In this case, we have
    \begin{equation*}
        |\bra{\lambda', T'} \cdot \ucg^\lambda \cdot \ket{d^*}\ket{T}|^2 =
        \frac{\prod_{i=1}^{d-1}(\lambda_1 - \mu_i + i - 1)}{\prod_{i=2}^{d}(
        \lambda_1 - \lambda_i + i - 1)},
\end{equation*}
where $\mu = \shape(T^{< d})$.
\end{lemma}

\subsection{An algorithm for quantum purity amplification}

The $k$-copy quantum purity amplification algorithm we study is essentially a straightforward generalization of the 1-copy-output algorithm from \cite{LFI+25}.
Given the input $\rho^{\otimes n}$,
the algorithm begins by performing weak Schur sampling and receiving a Young diagram $\blambda \vdash n$.
It then discards the permutation register, resulting in a state supported on the $Q_{\blambda}^d$ irrep.
The final step is for the algorithm to apply to the state a covariant quantum channel
\begin{equation*}
\Phi_{\blambda}^{(k)}:\calL(Q^d_{\blambda})\to \calL((\C^d)^{\otimes k}),
\end{equation*}
which we will define below.
The goal is for the output of this channel to be close to our desired state $\ket{v_d^{\otimes k}}$.

To define our channel, we will henceforth fix a Young diagram $\lambda \vdash n$.
Recall from the previous section the $k$-fold dual Clebsch--Gordan transform $\ucg^{\lambda, k}$, which satisfies the property that for any $U \in U(d)$,
\begin{equation}\label{eq:k-dual-cg}
    \ucg^{\lambda, k} \cdot \big(\overline{U}^{\otimes k} \otimes q_{\lambda}^d(U)\big) \cdot (\ucg^{\lambda,k})^{\dagger}
    = \bigoplus_{k\text{-paths } r} q_{\lambda(r)}^d(U).
\end{equation}
Recall also that when $\lambda_1 - \lambda_2 \geq k$, the $(\lambda - k e_1)$-irrep occurs on the right-hand side with multiplicity one.
Let us write $E_{\lambda - k e_1}$ for the orthogonal projector onto the unique copy of this irrep.
For notational convenience, we will write $\calH_{\mathsf{A}}$ for the Hilbert space $((Q^d_{(1)})^*)^{\otimes k}$
and $\calH_{\mathsf{B}}$ for the Hilbert space $Q_{\lambda}^d$.

\begin{definition}[QPA channel]
    If $\lambda_1 - \lambda_2 < k$, we define $\Phi_{\lambda}^{(k)}$ to output the maximally mixed state. Otherwise, when $\lambda_1 - \lambda_2 \geq k$, we define it as follows.
    Consider the matrix
\begin{equation}\label{eq:j-def}
    J_\lambda^{(k)}
    \coloneqq
    \frac{\dim(Q^d_\lambda)}{\dim(Q^d_{\lambda-ke_1})}\cdot (\ucg^{\lambda,k})^\dagger \cdot E_{\lambda-ke_1} \cdot \ucg^{\lambda,k},
\end{equation}
which acts on the space $\calH_{\mathsf{A}} \otimes \calH_{\mathsf{B}}$.
We will define $\Phi_{\lambda}^{(k)}$ to be the quantum channel whose Choi matrix is $J_{\lambda}^{(k)}$. In other words, for any~$X$ in $\calL(Q_{\lambda}^d)$, we have
\begin{equation}\label{eq:channel-def}
    \Phi_{\lambda}^{(k)}(X) = \tr_{\mathsf{B}}\Big(J_\lambda^{(k)} \cdot \big(I_{\mathsf{A}} \otimes X^T\big)\Big),
\end{equation}
where the transpose $X^T$ is taken with respect to the SSYT basis in $Q_{\lambda}^d$.
\end{definition}

Strictly speaking, as we have defined the QPA channel, the output is an element of $\big((Q^d_{(1)})^*\big)^{\otimes k}$, which we then identify with $(\C^d)^{\otimes k}$ via the map $\ket{i^*} \to \ket{i}$. It is easiest simply to choose our irrep's basis such that the dual basis vectors $\ket{i^*}$ coincide with the standard computational basis vectors $\ket{i}$.

As it will turn out, we will almost always be in the case where $\lambda_1 - \lambda_2 \geq k$,
and so it will be fine to use this assumption when designing our channel.
Conversely, we will view the case of $\lambda_1 - \lambda_2 < k$ as a rare failure mode, and as a result we are happy to set it to an arbitrary quantum channel.
Now we check that this is indeed a quantum channel.
First, we will show the following helper lemma.
For the rest of the subsection, we will use the following notational shorthand: $U_{\mathsf{A}} \coloneqq U^{\otimes k}$ and $U_{\mathsf{B}} \coloneqq q_{\lambda}^d(U)_{\mathsf{B}}$.

\begin{lemma}\label{lem:choi-commutes}
    For any unitary $U$,
    \begin{equation*}
        J_{\lambda}^{(k)}
        \cdot \big(U^*_{\mathsf{A}} \otimes U_{\mathsf{B}}\big)
        = \big(U^*_{\mathsf{A}} \otimes U_{\mathsf{B}}\big) \cdot J_{\lambda}^{(k)}.
    \end{equation*}
\end{lemma} 
\begin{proof}
    This is because
\begin{align*}
    (\ucg^{\lambda,k})^\dagger \cdot E_{\lambda-ke_1} \cdot \ucg^{\lambda,k}\cdot \big(U^*_{\mathsf{A}} \otimes U_{\mathsf{B}}\big)
    &={} (\ucg^{\lambda,k})^\dagger \cdot E_{\lambda-ke_1}\cdot \Big(\bigoplus_{k\text{-paths } r} q_{\lambda(r)}^d(U)\Big) \cdot \ucg^{\lambda,k}\\
    &={} (\ucg^{\lambda,k})^\dagger \cdot \Big(\bigoplus_{k\text{-paths } r} q_{\lambda(r)}^d(U)\Big)\cdot E_{\lambda-ke_1} \cdot \ucg^{\lambda,k} \\
    &={} \big(U^*_{\mathsf{A}} \otimes U_{\mathsf{B}}\big) \cdot (\ucg^{\lambda,k})^\dagger \cdot E_{\lambda-ke_1} \cdot \ucg^{\lambda,k},
\end{align*}
where we have applied Equation~\eqref{eq:k-dual-cg} twice.
\end{proof}

\begin{lemma}
$\Phi_{\lambda}^{(k)}$ is a valid quantum channel.
\end{lemma}
\begin{proof}
First, 
$J_{\lambda}^{(k)}$ is clearly positive semidefinite, and so $\Phi_{\lambda}^{(k)}$ is completely positive (see, e.g.\ \cite[Theorem~2.22]{Wat18}).
Thus, it remains to show that it is trace preserving.
Note that
\begin{align*}
    \tr_{\mathsf{A}}(J_{\lambda}^{(k)}) \cdot U_{\mathsf{B}}
    &= \tr_{\mathsf{A}}(J_{\lambda}^{(k)} \cdot U_{\mathsf{B}})\\
    &= \tr_{\mathsf{A}}(J_{\lambda}^{(k)} \cdot (U^* \cdot U^T)_{\mathsf{A}} \otimes U_{\mathsf{B}})\\
    &= \tr_{\mathsf{A}}((U^*_{\mathsf{A}} \otimes U_{\mathsf{B}}) \cdot J_{\lambda}^{(k)} \cdot U^T_{\mathsf{A}}) \tag{by \ref{lem:choi-commutes}}\\
    &= \tr_{\mathsf{A}}(U_{\mathsf{B}} \cdot J_{\lambda}^{(k)})\tag{by the cyclic property of the partial trace}\\
    &= U_{\mathsf{B}} \cdot \tr_{\mathsf{A}}(J_{\lambda}^{(k)}),
\end{align*}
and so $\tr_{\mathsf{A}}(J_{\lambda}^{(k)})$ commutes with $U_{\mathsf{B}}$, for any unitary $U$.
Thus,
Schur's lemma implies that
\begin{equation*}
\tr_{\mathsf{A}}(J_{\lambda}^{(k)})
= c \cdot I_{\mathrm{dim}(Q_{\lambda}^d)},
\end{equation*}
for some constant $c$.
To compute $c$, we take the trace of both sides, which says that $c \cdot \dim(Q_{\lambda}^d)$ is equal to
\begin{equation*}
    \mathrm{tr}(J_{\lambda}^{(k)})
    = \frac{\dim(Q^d_\lambda)}{\dim(Q^d_{\lambda-ke_1})}\cdot \tr\Big((\ucg^{\lambda,k})^\dagger \cdot E_{\lambda-ke_1} \cdot \ucg^{\lambda,k}\Big)
    = \frac{\dim(Q^d_\lambda)}{\dim(Q^d_{\lambda-ke_1})}\cdot \tr( E_{\lambda-ke_1})
    = \dim(Q_{\lambda}^d).
\end{equation*}
Thus, $c = 1$, and so $\tr_{\mathsf{A}}(J_{\lambda}^{(k)}) = I_{\dim(Q_{\lambda}^d)}$.
As a result, $\Phi_{\lambda}^{(k)}$ is trace preserving as well (see, e.g.\ \cite[Theorem~2.26]{Wat18}) and is therefore a quantum channel.
\end{proof}

Finally, we will need that our purity amplification algorithm is unitarily covariant.
This follows from the following lemma and the fact that the Schur transform is unitarily covariant.

\begin{lemma}\label{lem:unitarily-covariant}
    $\Phi_{\lambda}^{(k)}$ is unitarily covariant, meaning that for any unitary $U$,
    \begin{equation*}
        \Phi_{\lambda}^{(k)}(U_{\mathsf{B}} \cdot X \cdot U_{\mathsf{B}}^\dagger) = U_{\mathsf{A}} \cdot \Phi_{\lambda}^{(k)}(X) \cdot U_{\mathsf{A}}^\dagger.
    \end{equation*}
\end{lemma}
\begin{proof}
    Note that for any unitary $V$, $(V^\dagger)_{\mathsf{B}} = (V_{\mathsf{B}})^\dagger$,
    and because $q_{\lambda}^d$ is a polynomial irrep,
    $(V^*)_{\mathsf{B}} = (V_{\mathsf{B}})^*$,
    where the former complex conjugate is taken with respect to the standard basis and the latter is taken with respect to the SSYT basis of $Q_{\lambda}^d$;
    note that together these imply that \begin{equation*}(V^T)_{\mathsf{B}} = ((V^*)^{\dagger})_{\mathsf{B}}
    = ((V^*)_{\mathsf{B}})^{\dagger}
    = ((V_{\mathsf{B}})^*)^{\dagger}
    = (V_{\mathsf{B}})^T.
    \end{equation*}
    Using these facts, we have that
    \begin{align*}
        \Phi_{\lambda}^{(k)}(U_{\mathsf{B}} \cdot X \cdot U_{\mathsf{B}}^\dagger)
        &= \tr_{\mathsf{B}}\Big(J_\lambda^{(k)} \cdot \big(I_{\mathsf{A}} \otimes (U_{\mathsf{B}}^* \cdot X^T \cdot U_{\mathsf{B}}^T)\big)\Big)\tag{by \eqref{eq:channel-def}}\\
        &= \tr_{\mathsf{B}}\Big(J_\lambda^{(k)} \cdot \big((U \cdot U^\dagger)_{\mathsf{A}} \otimes (U_{\mathsf{B}}^* \cdot X^T \cdot U_{\mathsf{B}}^T)\big)\Big)\\
        &= \tr_{\mathsf{B}}\Big((U_{\mathsf{A}} \otimes U_{\mathsf{B}}^*) \cdot J_\lambda^{(k)} \cdot \big(U^\dagger_{\mathsf{A}} \otimes (X^T \cdot U_{\mathsf{B}}^T)\big)\Big)\tag{by \ref{lem:choi-commutes}}\\
        &= U_{\mathsf{A}} \cdot \tr_{\mathsf{B}}\Big((I_{\mathsf{A}} \otimes U_{\mathsf{B}}^*) \cdot J_\lambda^{(k)} \cdot \big(I^\dagger_{\mathsf{A}} \otimes (X^T \cdot U_{\mathsf{B}}^T)\big)\Big) \cdot U^\dagger_{\mathsf{A}}\\
        &= U_{\mathsf{A}} \cdot \tr_{\mathsf{B}}\Big(J_\lambda^{(k)} \cdot (I^\dagger_{\mathsf{A}} \otimes X^T)\Big) \cdot U^\dagger_{\mathsf{A}} \tag{by the cyclic property of the partial trace}\\
        &= U_{\mathsf{A}}\cdot \Phi_{\lambda}^{(k)}(X) \cdot U^\dagger_{\mathsf{A}}.\tag{by \eqref{eq:channel-def}}
    \end{align*}
    This completes the proof.
\end{proof}

\section{Proof of Theorem~\ref{thm:nonasymptotic-bounds}}

We will show 
\begin{equation} \label{eq:target_complexity}
    n \geq 12k + \frac{2032 + 4k}{\delta} \cdot \frac{(1-p_d)}{(p_d-p_{d-1})^2}
\end{equation}
copies suffice. That is, we may take $C = 2036$.

Let $\bw \sim p^{\otimes n}$, and let $\bh$ be its type. Set $(\blambda, \bS, \bT) = \rsk(\bw)$, and note that $(\blambda, \bS, \bT)$ is distributed as a sample from $\rsk^n(p)$.
Write $\bmu \coloneqq \shape(\bT^{< d})$.
Because $\bT$ is an SSYT, we have that $\bb_i \coloneqq \bmu_i - \blambda_{i+1} \geq 0$ for all $i$. We refer to this quantity as the $i$-th row's \emph{overhang}, as it measures the extent to which the non-$d$'s in the $i$-th row jut out beyond the $(i+1)$-st row beneath them. See Figure~\ref{fig:overhang}.




\begin{figure}[t!]
    \centering

    \newlength{\ytcell}
    \setlength{\ytcell}{2.4em}
    \ytableausetup{boxsize=\ytcell}

    \begin{tikzpicture}[baseline=(T.base)]
        \node[inner sep=0pt, outer sep=0pt] (T) {
            \begin{ytableau}
              1 & 1 & 1 & *(gray!30)1 & *(gray!30)2 & *(gray!70)3\\
              2 & *(gray!30)2 & *(gray!70)3 \\
              *(gray!70)3
            \end{ytableau}
        };

        \coordinate (NW) at (T.north west);
        \draw[
    decorate,
    decoration={brace, mirror, amplitude=3.5pt},
    line width=0.8pt
]
    ($(NW)+(3\ytcell+3pt,-\ytcell-2pt)$) --
    ($(NW)+(5\ytcell+1pt,-\ytcell-2pt)$)
    node[midway, below=5pt] {$b_1$};

        \draw[
            decorate,
            decoration={brace, mirror, amplitude=3.5pt},
            line width=0.8pt
        ]
            ($(NW)+(\ytcell+2pt,-2\ytcell-2pt)$) --
            ($(NW)+(2\ytcell-0.25pt,-2\ytcell-2pt)$)
            node[midway, below=5pt] {$b_2$};





    \end{tikzpicture}

    \caption{Illustration of the overhang $b_i \coloneqq \mu_i-\lambda_{i+1}$,
    shown for a fixed SSYT $T$ of shape \mbox{$\lambda=(6,3,1)$} over the
    alphabet $[3]$. Removing the darker gray boxes, which contain the letter
    $d=3$, leaves the sub-tableau $T^{<d}$ of shape $\mu=(5,2,0)$. The
    overhang $b_i$ counts the lighter gray boxes in row $i$: among the $\lambda_i - \lambda_{i+1}$ boxes in row $i$ extending beyond row $(i+1)$, the overhang counts those that are not labeled $d$. In this
    example diagram, we have $b_1=5-3=2$ and $b_2=2-1=1$. 
    }
    \label{fig:overhang}
\end{figure}

Because our purity amplification algorithm is unitarily covariant due to Lemma~\ref{lem:unitarily-covariant},
it suffices to prove this statement for the special case when $\rho$ is diagonal in the standard basis with increasing spectrum.
In this case, we saw from \eqref{eq:schur-rsk} that in the Schur basis, we can write
\begin{equation*}
    \schur \cdot \rho^{\otimes n} \cdot \schur^{\dagger}
    = \E_{(\blambda, \bS, \bT) \sim \rsk^n(p)} \big[\ketbra{\blambda, \bS, \bT}\big].
\end{equation*}
Since our algorithm for quantum purity amplification discards the symmetric group register before applying $\Phi^{(k)}_{\lambda}$, we have that
\begin{equation*}
    \qpa{k}(\rho^{\otimes n})
    = \E\big[\Phi_{\blambda}^{(k)}\big(\ketbra{\bT}\big)\big].
\end{equation*}
As a result, the output fidelity of the algorithm can be written as
\begin{equation*}
    \bra{d^{\otimes k}} \qpa{k}(\rho^{\otimes n}) \ket{d^{\otimes k}} = \E\big[F(\blambda, \bT)\big],
    \qquad\text{where}\qquad
    F(\blambda, \bT) \coloneqq \bra{d^{*\otimes k}}\Phi_{\blambda}^{(k)}\big(\ketbra{\bT}\big)\ket{d^{*\otimes k}}.
\end{equation*}
We note that on the left we are comparing the fidelity against the vector $\ket{d^{\otimes k}}$ because in the case of diagonal $\rho$, $\ket{v_d} = \ket{d}$. On the right, we are restoring the explicit dual basis for clarity in the lemma below.
Now we restate the output fidelity $F(\blambda, \bT)$ in terms of the combinatorics of Young tableaux.
To do so, we will make use of an expression for this fidelity from Honghao Fu's master's thesis~\cite{Fu16}, which we generalize to the case of $k > 1$.

\begin{lemma}[Combinatorial expression for the fidelity]\label{lem:combinatorial}
    Let $\lambda \vdash n$ and suppose $\lambda_1 - \lambda_2 \geq k$.
    Let $T$ be an SSYT of shape $\lambda$ and alphabet $[d]$.
    Set $\mu = \shape(T^{< d})$ and $b_i = \mu_i - \lambda_{i+1}$.
    For each $i > 1$, set $\Delta_i \coloneqq \lambda_1 - \lambda_i + i - 2$. Then
    \begin{equation}\label{eq:our-fidelity-formula}
        F(\lambda,T) = \prod_{i=2}^{d} \frac{(\Delta_i - b_{i-1})^{\downarrow k}}{(\Delta_i)^{\downarrow k}}.
    \end{equation}
\end{lemma}
\begin{proof}
We begin by expressing the fidelity $F(\lambda, T)$ as 
\begin{align}
\bra{d^{*\otimes k}}\Phi_{\lambda}^{(k)}\big(\ketbra{T}\big)\ket{d^{*\otimes k}}
&= \bra{d^{*\otimes k}}\tr_{Q_{\lambda}^d}\Big(J_\lambda^{(k)} \cdot \big(I \otimes (\ketbra{T})^T\big)\Big)\ket{d^{*\otimes k}}
    \tag{by \eqref{eq:channel-def}}\\
&= \bra{d^{*\otimes k}}\tr_{Q_{\lambda}^d}\Big(J_\lambda^{(k)} \cdot \big(I \otimes \ketbra{T}\big)\Big)\ket{d^{*\otimes k}}\nonumber\\
&= \bra{d^{*\otimes k}} \bra{T}J_{\lambda}^{(k)}\ket{d^{*\otimes k}} \ket{T}\nonumber\\
&= \frac{\dim(Q^d_\lambda)}{\dim(Q^d_{\lambda-ke_1})}\cdot \bra{d^{*\otimes k}} \bra{T}\cdot (\ucg^{\lambda,k})^\dagger \cdot E_{\lambda-ke_1} \cdot \ucg^{\lambda,k}\cdot \ket{d^{*\otimes k}} \ket{T}\tag{by \eqref{eq:j-def}}\\
&= \frac{\dim(Q^d_\lambda)}{\dim(Q^d_{\lambda-ke_1})}\cdot \Vert E_{\lambda-ke_1} \cdot \ucg^{\lambda,k}\cdot \ket{d^{*\otimes k}} \ket{T}\Vert^2.\label{eq:goal-value}
\end{align}

To understand the length of this vector, by \eqref{eq:induction} we have
\begin{align}
    E_{\lambda-ke_1} \cdot \ucg^{\lambda,k}\cdot \ket{d^{*\otimes k}} \ket{T}
    &= E_{\lambda-ke_1} \cdot \ucg^{\lambda,k-1 \rightarrow k}\cdots \ucg^{\lambda,1 \rightarrow 2} \ucg^{\lambda}\cdot \ket{d^{*\otimes k}} \ket{T}\nonumber\\
    &= \sum_{T^k} \ketbra{1^k, T^k} \cdot \ucg^{\lambda,k-1 \rightarrow k}\cdots \ucg^{\lambda,1 \rightarrow 2} \ucg^{\lambda}\cdot \ket{d^{*\otimes k}} \ket{T},\label{eq:split-the-ucg}
\end{align} 
where the sum is over all SSYTs $T^k$ of shape $\lambda(1^k) = \lambda - k e_1$.
Recall that applying the dual Clebsch--Gordan transform removes an additional box from the Young diagram specified by the current path.
Reversing this, $\bra{1^k, T^k} \cdot \ucg^{\lambda,k-1 \rightarrow k}$ must be supported on vectors of the form $\bra{1^{k-1}, T^{k-1}}$, and so
\begin{equation*}
    \bra{1^k, T^k} \cdot \ucg^{\lambda,k-1 \rightarrow k}
    = \sum_{T^{k-1}} \bra{1^k, T^k} \cdot \ucg^{\lambda,k-1 \rightarrow k} \cdot \ketbra{1^{k-1},T^{k-1}}.
\end{equation*}
Applying this logic repeatedly, we have
\begin{equation}\label{eq:big-prod}
    \eqref{eq:split-the-ucg}
    = \sum_{T^1, \ldots, T^k} \ket{1^k, T^k} \cdot \prod_{t=0}^{k-1} \bra{1^{t+1},T^{t+1}} \ucg^{\lambda, t \rightarrow t+1} \ket{d^*}\ket{1^t, T^t},
\end{equation}
where the sum is over all SSYTs $T^t$ of shape $\lambda(1^t) = \lambda - t e_1$,
and we have set $T^0 \coloneqq T$.
Now, by \eqref{eq:one-step-cg}, we have that
\begin{equation*}
    \bra{1^{t+1}, T^{t+1}} \ucg^{\lambda, t\rightarrow t+1} \ket{d^*}\ket{1^t, T^t} = \bra{1^{t+1}, T^{t+1}} \ucg^{\lambda(1^t)} \ket{d^*}\ket{T^t}.
\end{equation*}
As a result, we can rewrite the above expression as
\begin{equation*}
\eqref{eq:big-prod}=
    \sum_{T^1, \ldots, T^k} \ket{1^k, T^k} \cdot \prod_{t=0}^{k-1} \bra{1^{t+1}, T^{t+1}} \ucg^{\lambda(1^t)} \ket{d^*}\ket{T^t}.
\end{equation*}
Finally, by Lemma~\ref{lem:cg-coeff}, we know that $\bra{1^{t+1}, T^{t+1}} \ucg^{\lambda(1^t)} \ket{d^*}\ket{T^t} = 0$  unless the final box in the first row of $T^t$ contains the letter ``$d$'', and $T^{t+1}$ is the SSYT which results from removing this box.
This implies that \eqref{eq:big-prod} is actually 0 if $T$ does not have at least $k$ boxes containing the letter ``$d$'' at the end of its first row;
however, note that in this case, our desired bound of~\eqref{eq:our-fidelity-formula} is also equal to 0, because 
\begin{equation*}
    \Delta_2 - b_1
    = (\lambda_1 - \lambda_2) - (\mu_1 - \lambda_2)
    = \lambda_1 - \mu_1
    = (\text{the number of $d$'s in $T$'s first row}) < k,
\end{equation*}
and so $(\Delta_2 - b_1)^{\downarrow k} = 0$.
Thus,
we will henceforth assume
that $T$ has at least $k$ boxes containing the letter ``$d$'' at the end of its first row;
in this case, we now have that
\begin{equation*}
    E_{\lambda-ke_1} \cdot \ucg^{\lambda,k}\cdot \ket{d^{*\otimes k}} \ket{T}
    = \ket{1^k, T^k} \cdot \prod_{t=0}^{k-1} \bra{1^{t+1}, T^{t+1}} \ucg^{\lambda(1^t)} \ket{d^*}\ket{T^t},
\end{equation*}
where we have fixed $T^t$ to be the SSYT resulting from removing the last $t$ boxes in the first row of $T$.
Thus,
\begin{equation}\label{eq:squared-length}
    \Vert E_{\lambda-ke_1} \cdot \ucg^{\lambda,k}\cdot \ket{d^{*\otimes k}} \ket{T}\Vert^2
    = \prod_{t=0}^{k-1} |\bra{1^{t+1}, T^{t+1}} \ucg^{\lambda(1^t)} \ket{d^*}\ket{T^t}|^2.
\end{equation}
To apply Lemma~\ref{lem:cg-coeff}, let $\mu = \shape(T^{< d})$, and note that because only $d$'s are removed when going from $T$ to any of the $T^t$'s, $\mu$ is also equal to $\shape((T^t)^{<d})$ for any $t$. Thus,
\begin{equation*}
    \eqref{eq:squared-length}
    = \prod_{t=0}^{k-1} \frac{\prod_{i=1}^{d-1}(\lambda_1-t - \mu_i + i - 1)}{\prod_{i=2}^{d}(
        \lambda_1-t - \lambda_i + i - 1)}
    = \frac{\prod_{i=1}^{d-1}(\lambda_1 - \mu_i + i - 1)^{\downarrow k}}{\prod_{i=2}^{d}(
        \lambda_1 - \lambda_i + i - 1)^{\downarrow k}}
         =\prod_{i=2}^{d} \frac{(\Delta_i - b_{i-1})^{\downarrow k}}{(
        \Delta_i + 1)^{\downarrow k}},
\end{equation*}
where in the last line we used the fact that $b_i = \mu_i - \lambda_{i+1}$, and so 
\begin{equation*}
\lambda_1 - \mu_i + i - 1 = \lambda_1 - \lambda_{i+1} - b_i + i -1
= \Delta_{i+1} - b_i.
\end{equation*}
This concludes our computation of the length of the vector.

Revisiting our goal of computing \eqref{eq:goal-value}, it remains to compute the ratio of the dimensions.
To do so, let us recall the Weyl dimension formula, which states that
\begin{equation*}
    \dim (Q^d_\nu)
    =\prod_{1\leq i<j\leq d}\frac{\nu_i-\nu_j+j-i}{j-i}.
\end{equation*}
When applying this to $\nu=\lambda$ and $\nu=\lambda-ke_1$, all factors with $i\neq 1$
match. Hence,
\begin{equation*}
    \frac{\dim (Q^d_\lambda)}{\dim(Q^d_{\lambda-ke_1})}
    =
    \prod_{i=2}^d
    \frac{\lambda_1-\lambda_i+i-1}{\lambda_1-k-\lambda_i+i-1}
    =
    \prod_{i=2}^d
    \frac{\Delta_i+1}{\Delta_i-k+1}.
\end{equation*}
In summary, we have
\begin{equation*}
    \bra{d^{\otimes k}}\Phi_{\lambda}^{(k)}\big(\ketbra{T}\big)\ket{d^{\otimes k}}
    = \prod_{i=2}^d
    \frac{\Delta_i+1}{\Delta_i-k+1} \cdot \frac{(\Delta_i - b_{i-1})^{\downarrow k}}{(
        \Delta_i + 1)^{\downarrow k}}
    = \prod_{i=2}^d
    \frac{(\Delta_i - b_{i-1})^{\downarrow k}}{(
        \Delta_i)^{\downarrow k}}.
\end{equation*}
This completes the proof.
\end{proof}

Let us first manipulate the expression that Lemma~\ref{lem:combinatorial} produces in order to understand it better. Note that
\begin{align*}
   \frac{(\bDelta_i - \bb_{i-1})^{\downarrow k}}{(\bDelta_i)^{\downarrow k}}
   &= \frac{(\bDelta_i - \bb_{i-1}) \cdots (\bDelta_i - \bb_{i-1}-k+1)}{ \bDelta_i \cdots (\bDelta_i - k+1)}\\
   &= \prod_{j=0}^{k-1} \frac{\bDelta_i - \bb_{i-1} - j}{\bDelta_i - j}
   = \prod_{j=0}^{k-1} \Big(1 - \frac{\bb_{i-1}}{\bDelta_i - j}\Big)
   \geq \Big(1 - \frac{\bb_{i-1}}{\bDelta_i - k+1}\Big)^k
   \geq 1 - k \cdot \Big(\frac{\bb_{i-1}}{\bDelta_i -k+1}\Big).
\end{align*}
Hence,
\begin{equation*}
    \prod_{i=2}^{d} \frac{(\bDelta_i - \bb_{i-1})^{\downarrow k}}{(\bDelta_i)^{\downarrow k}}
    \geq \prod_{i=2}^d \Big(1 - k \cdot \Big(\frac{\bb_{i-1}}{\bDelta_i-k+1}\Big)\Big)
    \geq 1 - \sum_{i=2}^d k \cdot \Big(\frac{\bb_{i-1}}{\bDelta_i-k+1}\Big).
\end{equation*}
Note that for $i \geq 2$, $\bDelta_i = \blambda_1 - \blambda_i + i - 2 \geq \blambda_1 - \blambda_2$. As a result, so long as $\blambda_1 - \blambda_2\geq k$, we have that
\begin{equation}\label{eq:fidelity-lb}
    F(\blambda, \bT) \geq 1 - \Big(\frac{k}{\blambda_1 - \blambda_2 - k + 1}\Big) \cdot \sum_{i=2}^d \bb_{i-1}.
\end{equation}
We want to show that the fidelity is close to 1, which entails showing that the second term is small. 
This term is the ratio of two expressions which are correlated random variables depending on the sample $(\blambda, \bS, \bT)$ from the RSK distribution.
To decouple these random variables, we will choose an event $\bcalE$ conditioned on which the denominator essentially becomes a fixed value. Doing so will allow us to analyze the numerator by itself, which is an easier task.

Typically, we expect that $\blambda_1$ should be roughly equal to $p_d \cdot n$ and $\blambda_2$ should be roughly equal to $p_{d-1} \cdot n$.
So we expect that $\blambda_1 - \blambda_2$ should be at least, say, $\frac{1}{2}(p_d - p_{d-1}) \cdot n$ with reasonably large probability.
Thus, let us define $\bcalE$ to be the event that $\blambda_1 - \blambda_2 \geq \frac{1}{2}(p_d - p_{d-1}) \cdot n$.
By our choice of $n$, we have 
\begin{equation}
    \frac{1}{2}(p_d - p_{d-1}) \cdot n \geq 2k. \label{eq:at_least_2k}
\end{equation}
To see this, we consider two cases. First, if $p_d - p_{d-1} \geq 1/3$, then the bound holds since $n \geq 12k$ as well. Second, if $p_d - p_{d-1} \leq 1/3$, then it must be that $p_d \leq 2/3$, in which case 
\begin{equation*}
    \frac{1}{2}(p_d - p_{d-1}) \cdot n \geq \frac{1}{2} (p_d-p_{d-1}) \cdot \Big( \frac{4k}{\delta} \cdot \frac{(1-p_d)}{(p_d-p_{d-1})^2}\Big) \geq 2k \cdot \Big(\frac{1-p_d}{p_d-p_{d-1}}\Big) \geq 2k
\end{equation*}
where we have used $\delta \leq 1$, and $1-p_d \geq 1/3 \geq p_d - p_{d-1}$. Thus, for our choice of $n$, \eqref{eq:at_least_2k} holds. Hence, when the event $\calE$ occurs, $\blambda_1 - \blambda_2 \geq k$, which we recall was sufficient for the lower bound in Equation~\eqref{eq:fidelity-lb} to hold.
In addition, we have that $\blambda_1 - \blambda_2 -k + 1 \geq \frac{1}{4}(p_d - p_{d-1}) \cdot n$.
Plugging this into Equation~\eqref{eq:fidelity-lb}, when $\bcalE$ occurs, we have that
\begin{equation}\label{eq:final-lb}
    F(\blambda, \bT) \geq 1 - \Big(\frac{4k}{(p_d-p_{d-1})\cdot n}\Big) \cdot \sum_{i=1}^{d-1} \bb_{i}.
\end{equation}
Having done this,
we see that the denominator is now fixed,
and this will make analyzing the numerator much simpler.

Now, averaging over all $(\blambda, \bS, \bT)$,
we can lower bound the output fidelity by
\begin{align*}
    \bra{v_d^{\otimes k}} \qpa{k}(\rho^{\otimes n}) \ket{v_d^{\otimes k}}
    &\geq \E[\bone[\blambda_1 - \blambda_2 \geq k] \cdot F(\blambda, \bT)]\\
    &\geq \E[\bone[\bcalE] \cdot F(\blambda, \bT)] \tag{because $\bcalE \Rightarrow \blambda_1 -\blambda_2 \geq k$}\\
    &\geq \E\Big[\bone[\bcalE]\cdot\Big(1 - \Big(\frac{4k}{(p_d-p_{d-1})\cdot n}\Big) \cdot \sum_{i=2}^d \bb_{i-1}\Big)\Big] \tag{by \eqref{eq:final-lb}}\\
    &\geq \Pr[\bcalE] - \Big(\frac{4k}{(p_d-p_{d-1})\cdot n}\Big) \cdot \E\Big[\sum_{i=1}^{d-1} \bb_{i}\Big].
\end{align*}
As argued before, the probability $\Pr[\bcalE]$ should be reasonably large, and we will argue below that it is indeed close to~1. Before doing so, however, let us focus on the sum of the $\bb_i$'s. Expanding it, 
\begin{equation*}
    \sum_{i=1}^{d-1} \bb_i
    = \sum_{i=1}^{d-1} \bmu_i - \sum_{i=1}^{d-1} \blambda_{i+1}
    = |\bmu| - (|\blambda| - \blambda_1)
    = |\bmu| - n + \blambda_1.
\end{equation*}
Now, $|\bmu|$ is equal to the number of non-$d$ letters in $\bw$, and is therefore $(p_1 + \cdots + p_{d-1}) \cdot n$ in expectation.
Combining this with Lemma~\ref{lem:first-row}, we have that
\begin{equation*}
    \E\Big[\sum_{i=1}^{d-1} \bb_{i}\Big]
    = (p_1 + \cdots + p_{d-1})\cdot n - n + \E[\blambda_1]
    = \E[\blambda_1] - p_d \cdot n
    \leq \sum_{i=1}^{d-1} \frac{p_i}{p_d - p_i}
    \leq \sum_{i=1}^{d-1} \frac{p_i}{p_d - p_{d-1}}
    = \frac{(1-p_d)}{p_d-p_{d-1}}.
\end{equation*}
Hence,
\begin{equation}\label{eq:almost-there}
\bra{v_d^{\otimes k}} \qpa{k}(\rho^{\otimes n}) \ket{v_d^{\otimes k}}
\geq \Pr[\bcalE] -4 \cdot \frac{k}{n} \cdot \frac{1-p_d}{(p_d-p_{d-1})^2}.
\end{equation}
It remains to bound $\Pr[\bcalE]$,
which we do in the following lemma.
Plugging this bound into \eqref{eq:almost-there} gives
\begin{equation*}
    \bra{v_d^{\otimes k}} \qpa{k}(\rho^{\otimes n}) \ket{v_d^{\otimes k}} \geq 1 - \frac{2032+4k}{n} \cdot \frac{(1-p_d)}{(p_d-p_{d-1})^2} \geq 1-\delta,
\end{equation*}
by our choice of $n$ in \eqref{eq:target_complexity}. This completes the proof.

\begin{lemma}[Concentration of the denominator]
\begin{equation*}
    \Pr\big[\,\overline{\bcalE}\,\big] \;\leq\; \frac{2032}{n} \cdot \frac{(1-p_d)}{(p_d-p_{d-1})^2}.
\end{equation*}
\end{lemma}
\begin{proof}
Let us recall that $\overline{\bcalE}$ is the event that $\blambda_1 - \blambda_2 < \frac{1}{2}(p_d - p_{d-1}) \cdot n$.
When this occurs, it must be the case that either
\begin{equation*}
    (i)~\blambda_1 \leq p_d \cdot n - \frac{1}{4}(p_d - p_{d-1}) \cdot n,
    \qquad\text{or}\qquad
    (ii)~\blambda_2 \geq p_{d-1} \cdot n + \frac{1}{4}(p_d - p_{d-1}) \cdot n.
\end{equation*}
We will bound the probability of these events separately, and then our bound on $\overline{\bcalE}$ will follow from a union bound.

For event (i), let us note that $\bh_d$, the number of $d$'s in $\bw$, is at most $\blambda_1$. Hence, when event (i) occurs, it must be the case that
\begin{equation*}
    \bh_d \leq p_d \cdot n - \frac{1}{4}(p_d - p_{d-1}) \cdot n.
\end{equation*}
But $\bh_d$ is distributed as $\mathrm{Binomial}(n, p_d)$ and hence has variance $p_d (1- p_d)\cdot n$. Thus, by Chebyshev's inequality,
\begin{equation}\label{eq:bound-on-event-i}
    \Pr\Big[\bh_d \leq p_d \cdot n - \frac{1}{4}(p_d - p_{d-1}) \cdot n\Big]
    \leq \frac{p_d(1-p_d)\cdot n}{(\frac{1}{4}(p_d - p_{d-1})\cdot n)^2}
    \leq \frac{16}{n} \cdot \frac{(1-p_d)}{(p_d-p_{d-1})^2}.
\end{equation}
Hence, this bounds the probability that event (i) occurs as well.

For event (ii), let us write $\bX \coloneqq \blambda_2 - p_{d-1} \cdot n$. Then Lemma~\ref{lem:second-row} implies that
\begin{equation*}
    \E \big[\bX^2\big] \leq 84 p_{d-1} \cdot n +42 (1 - p_d) \cdot n \leq 126 (1-p_d) \cdot n.
\end{equation*}
Thus, Markov's inequality implies that
\begin{multline*}
    \Pr\Big[\blambda_2 \geq p_{d-1} \cdot n + \frac{1}{4}(p_d - p_{d-1}) \cdot n\Big]
    = \Pr\Big[\bX \geq \frac{1}{4}(p_d - p_{d-1}) \cdot n\Big]\\
    \leq \Pr\bigg[\bX^2 \geq \Big(\frac{1}{4}(p_d - p_{d-1}) \cdot n\Big)^2\bigg]
    \leq \frac{126(1-p_d) \cdot n}{(\frac{1}{4}(p_d - p_{d-1}) \cdot n)^2}
    = \frac{2016}{n} \cdot \frac{(1-p_d)}{(p_d-p_{d-1})^2}.
\end{multline*}
Combining this with \eqref{eq:bound-on-event-i} completes the proof.
\end{proof}

\bibliographystyle{alpha}
\bibliography{wright}

\end{document}